\documentclass[conference]{IEEEtran}
\IEEEoverridecommandlockouts
\usepackage{cite}
\usepackage{amsmath,amssymb,amsfonts}
\usepackage{algorithmic}
\usepackage{graphicx}
\usepackage{textcomp}
\usepackage{xcolor}
\def\BibTeX{{\rm B\kern-.05em{\sc i\kern-.025em b}\kern-.08em
    T\kern-.1667em\lower.7ex\hbox{E}\kern-.125emX}}
\usepackage[ansinew]{inputenc}
\usepackage{amsthm}
    
\newtheorem{theorem}{Theorem}
\newtheorem{proposition}[theorem]{Proposition}
\newtheorem{lemma}[theorem]{Lemma}
\newtheorem{definition}[theorem]{Definition}
\newtheorem{corollary}[theorem]{Corollary}
\newtheorem{remark}[theorem]{Remark}

\newcommand{\supp}{\mathrm{supp}}
\newcommand{\LT}{\mathrm{LT}}
\newcommand{\LP}{\mathrm{LP}}
\newcommand{\Zs}{\Z^s_{\geq 0}}

\newcommand{\Q}{\mathcal{Q}}
\newcommand{\D}{ \mathcal{D}}
\newcommand{\G}{ \mathcal{G}}

\newcommand{\N}{{\mathbb N}}
\newcommand{\Z}{{\mathbb Z}}

\newcommand{\F}{{\mathbb F}}
\renewcommand{\L}{{\mathbb L}}

\newcommand{\X}{ \boldsymbol{X}}
\newcommand{\bs}[1]{\boldsymbol{#1}}
\newcommand{\I}{ \bs{\mathcal{I}}}

\newcommand{\flecha}{\rightarrow}

\newcommand{\tq}{\;\mid\;}

\begin{document}

\title{A Note on the Theoretical Support to Compute Dimension in Abelian Codes
\thanks{Support: Grant PID2020-113206GB-I00 funded by MICIU/ AEI/  10.13039/501100011033 y Fundaci\'{o}n S\'{e}neca de Murcia, project 22004/PI/22.}
}

\author{\IEEEauthorblockN{Jos\'e Joaqu\'{\i}n Bernal and Juan Jacobo Sim\'on}
\IEEEauthorblockA{\textit{Departamento de Matem\'aticas} \\
\textit{Universidad de Murcia}\\
Murcia, Spain \\
\{josejoaquin.bernal, jsimon\}@um.es}}

\maketitle

\begin{abstract}
In this note we give a theoretical support by means of quotient polynomial rings for the computation formulas of the dimension of abelian codes.
\end{abstract}


\section{Introduction}

As it is known, the computation of the dimension of abelian codes may be done by easy formulas expressed in terms of their defining sets (see Theorem~\ref{dimension codigos conj def} below).Those formulas was developed in the setting of group rings (see \cite{elias} for a nice exposition) and all proofs are written by means of character theory. This situation comes from the fact that the original frame of abelian codes were abelian group rings and group representations (see \cite{camion}). 

Over time, more complicated computations and tools drove a number of theorists to consider abelian codes in the setting of quotient rings of polynomial rings. In this way, the theoretical support of some basic computations, as dimension in terms of defining sets, has not been updated, or has been only partially updated, as in \cite{Poli}.

Another way to find the dimension of an abelian code may be done by using Groebner basis, as in \cite{Cox}. It may be done in terms of what is called the footprint of an ideal (see \cite{blahut}). All computations and theoretical results are, by default, expressed in terms of polynomials; and, in contrast with the case of defining sets, all arguments are clearly known.

The aim of this note is to give a full theoretical support for the computation formulas to get the dimension of abelian codes. First, we will study the correspondence between abelian codes and their defining sets. Then we shall deduce the computation formulas in those terms. Then, we recall the definition and construction of footprints and see the computation formula that arises from them. Finally, we comment on some interactions between defining sets and footprints from Groebner basis.

\section{Notation and basic definitions}\label{seccion de notacion}

Throughout this note, $\F$  will be a finite field with $q$ elements, with $q$ a power of a prime number, $r_i$ will be positive integers, for $i\in \{ 1,\dots,s\}$, and $n=r_1\cdots r_s$.  We denote by $\Z_{r_i}$ the ring of integers modulo $r_i$. We always write its elements as \textbf{canonical representatives}. When necessary, we write $\overline{a}\in\Z_k$ for any $a\in\Z$ and $k\in \N$. For any commutative ring $R$ and any subset $A\subseteq R$ we denote by $\langle A\rangle$ the ideal generated by $A$, as usual.
 
An \textbf{Abelian Code} of length $n=r_1\dots r_s$ is an ideal in the algebra   $\F(r_1,\ldots, r_s)=\F[X_1,\ldots, X_s]/\langle X_1^{r_1}-1,\ldots, X_s^{r_s}-1\rangle$ and we assume that this algebra is semisimple; that is, $\gcd (n,q)=1$. For the sake of simplicity, we often refer to ``Abelian Code'' simply as ``code''. 
It is well-known (see, for example, \cite{anderson-fuller}) that every finite semisimple commutative ring has a (unique) \textbf{complete set of primitive orthogonal idempotents} associated to its decomposition as finite product of simple rings.

Codewords are identified with polynomials. We  denote the weight of a codeword $c$ by $\omega(c)$. We set $\I=\Z_{r_1}\times\cdots\times \Z_{r_s}$ and we  write the elements $f \in  \F(r_1,\dots,r_s)$ as its canonical representatives; that is, $f=\sum a_m \X^m$, where $m=(m_1,\dots, m_s)\in \I$ and $\X^m=X_1^{m_1}\cdots X_s^{m_s}$. When necessary, for a polynomial $f \in \F[X_1,\dots,X_s]$ we denote by $\overline{f}$ its image under the canonical projection onto $\F(r_1,\dots,r_s)$. 

For each $i\in \{ 1,\ldots, s\}$, we denote by $R_{r_i}$ (resp. $U_{r_i}$) the set of all $r_i$-th roots of unity (resp. all $r_i$-th primitive roots of unity) and define $R=\prod_{i=1}^s R_{r_i}$ ($U=\prod_{i=1}^s U_{r_i}$). For $m=(m_1,\ldots, m_s)\in \I$, we write $\bs{\alpha}^{m} = (\alpha_1^{m_1}, \dots ,\alpha_s^{m_s})$. Throughout this paper, we fix the notation $\L$ for an extension field, $\L|\F$, containing $U_{r_i}$, for all  $i\in\{1,\dots,s\}$.

For a canonical representative $f \in  \F(r_1,\dots,r_s)$ seen as polynomial $f=f(X_1,\dots,X_s) \in \F[X_1,\dots,X_s]$  and $\bs{\alpha}\in R$, we write $f(\bs{\alpha})=f(\alpha_1,\dots,\alpha_s)$.

\section{Defining sets in Abelian Codes}

The study of several important parameters of the Cyclic and Abelian Codes is carried out through their definition sets. For example, their minimum distance or dimension may be computed easily by means of them and it is also essential in all decoding techniques because they are needed to determine the correction capability and to form the syndrome values of  error polynomials.

In this section we shall prove that every abelian code $C$ in $\F(r_1,\dots,r_s)$ is totally determined by its defining set (see definition below). We keep all notation above; in particular, we recall that $\F(r_1,\dots,r_s)$ is a semisimple ring.

%\begin{definition}%\label{clase q-ciclotomica}
 We recall the notion of cyclotomic coset. For any $\gamma\in\N$ the $q^{\gamma}$-cyclotomic coset of an integer $a$ modulo $r$ is the set
\[ C_{(q^\gamma,r)}(a)=\left\{a\cdot q^{\gamma\cdot i} \tq i \in \N\right\} \subseteq \Z_r.\]
%\end{definition}

By elementary algebraic properties of polynomials over finite fields it is well-known that if $(a,\dots,a_s)$ is a root of the polynomial $f\in \F[\X]$ then $(a_1^q,\dots,a_s^q)$ is also a root of $f$. This property is the key of the extension of the notion of $q$-cyclotomic coset to several variables.

 \begin{definition}\label{qorbita}
In the setting described above, given an element $(a_1,\dots,a_s)\in I$, we define its $q$-orbit modulo  $\left(r_1,\dots,r_s\right)$ as the subset $Q(a_1,\dots,a_s)= \left\{\left(a_1\cdot q^i ,\dots,a_s\cdot q^i  \right)\tq i\in \N\right\} \subseteq \I$.
\end{definition}

We also extend the notion of set of zeros for Cyclic Codes.

\begin{definition}
 Let $A\subset \F(r_1,\dots,r_s)$. The set of zeros of $A$ is $\mathcal{Z}(A)=\left\{\bs{\beta}\in R \tq f(\bs\beta)=0,\text{ for all } f\in A\right\}$.
\end{definition}

Once one has fixed $\bs\alpha\in U$ one may define the following object.

\begin{definition}\label{conjunto de definicion}
 In the setting above, the defining set of $A$, with respect to $\bs\alpha\in U$ is $\D_{\bs\alpha}(A)=\left\{m\in \I\tq \bs\alpha^m\in \mathcal{Z}(A)\right\}$.
\end{definition}

\begin{remark}\rm{
 Let $C$ be a code in $\F(r_1,\dots,r_s)$ and $\bs\alpha\in U$. From Definition~\ref{qorbita} we have immediately that the defining set $\D_{\bs\alpha}(C)$ is a disjoint union of $q$-orbits.}
\end{remark}

Now we shall prove that there is a one to one correspondence between Abelian Codes and the sets of $q$-orbits.

\begin{lemma}[See \cite{ghorpade}]\label{polinomio cero caracterizado}
	Let $f\in \F(r_1,\dots,r_s)$ be an arbitrary element and consider $\bs{\alpha}\in U$. If $f(\bs{\alpha}^m)=0$ for all $m\in \I$ then $f=0$ in $\F[\X]$.
\end{lemma}
\begin{proof}
	We set $\bs \alpha =(\alpha_1,\dots,\alpha_s)$ and we proceed by induction on $s$:
	
	For $s=1$ is obvious by classical arguments on degrees of polynomials in one variable. Suppose the result is true for $s-1$ and we have to see for $s$. 
	
	Let $f(\X) = \sum_{j=0}^{r_s-1}f_j(X_1,\dots,X_{s-1})X_s^j$ be such that $f(\bs{\alpha}^m)=0$ for all $m\in \I$. If for each $m \in \I$ we set  $f_{\alpha_1^{m_1},\dots, \alpha_{s-1}^{m_s-1}}= \sum_{j=0}^{r_s-1}f_j(\alpha_1^{m_1},\dots, \alpha_{s-1}^{m_s-1})X_s^j$, we see that it has the roots
	$\alpha_s^0,\dots, \alpha_s^{r_s-1}$ ($r_s$-roots) and it has degree at most $r_s-1$; so that $f_{\alpha_1^{m_1},\dots, \alpha_{s-1}^{m_s-1}}= 0$. This means that $f_j(X_1,\dots,X_{s-1})$ vanishes in $\bs{\alpha}^m$ for all $j = 0,\dots,s-1$. Hence, by induction hypothesis, $f_j(X_1,\dots,X_{s-1}) = 0$ and then $f(\X) = \sum_{j=0}^{r_s-1} 0 \cdot X_s^j = 0$.
\end{proof}

We denote by $\Q$ the collection of all subsets of $\I$ that are unions of $q$-orbits. We consider the usual partial ordering in $\Q$, given by set inclusion. Note that if $r_i\geq 2$ for all $i=1,\dots,s$ then $|\Q|\geq 3$, as $Q(0,0)$, $Q(1,0)$ and $Q(0,1)$ are disjoint sets.

First, we settle the correspondence $C\mapsto \D_{\bs\alpha}(C)\in \Q$. We begin by listing some properties that may be proved in a direct way.

\begin{proposition}\label{propiedades con def inclusion y suma}
	Let $C$ and $D$ codes in $\F(r_1,\dots,r_s)$, with defining sets $\D_{\bs\alpha}(C)$ and $\D_{\bs\alpha}(D)$, respectively. Then
	\begin{enumerate}
		\item If $C\subseteq D$ then $\D_{\bs\alpha}(D)\subseteq \D_{\bs\alpha}(C)$.
		\item  $\D_{\bs\alpha}(C + D)=\D_{\bs\alpha}(C)\cap\D_{\bs\alpha}(D)$.
	\end{enumerate} 
\end{proposition}

As a consequence we have the following corollary.

\begin{corollary}\label{suma de primitivos interseccion de conjuntos def}
	Let $e_1,\dots,e_t$ be the complete set of primitive orthogonal idempotents of $\F(r_1,\dots,r_s)$, and write $C=\sum_{k=1}^l Re_{i_k}$. Then $\D_{\bs\alpha}(C)=\cap_{k=1}^l\D_{\bs\alpha}(Re_{i_k})$
\end{corollary}

Note that the intersection of all defining sets is  $\cap_{k=1}^n\D_{\bs\alpha}(Re_{i_k})=\emptyset$. 

Now we shall prove injectivity. First we prove the following lemma whose proof is straightforward.

\begin{lemma}\label{igualdad conj de definicion}
	Let $C$ be a code in $\F(r_1,\dots,r_s)$ and $\bs\alpha\in U$. Consider the ideal (code)
	\[D=\{f\in \F(r_1,\dots,r_s)\tq f(\bs\alpha^m)=0,\;\text{for all }m\in \D_{\bs\alpha}(C)\}.\]
	Then $\D_{\bs\alpha}(C)=\D_{\bs\alpha}(D)$
\end{lemma}

Next lemma is crucial to prove injectivity.

\begin{lemma}\label{De es maximal}
	Let $e\in \F(r_1,\dots,r_s)$ be a primitive idempotent. Then $\D_{\bs\alpha}(Re)$ is maximal in $\Q$.
	
	Conversely, if $\I\neq Q\in \Q$ is a maximal element, then $C=\{f\in \F(r_1,\dots,r_s)\tq f(\bs\alpha^m)=0,\;\text{for all }m\in Q\}$ is an ideal such that $C=Re$ with $e$ a primitive idempotent.
\end{lemma}

\begin{proof} Let $e_1,\dots,e_t$ be the complete set of orthogonal idempotents of the ring $\F(r_1,\dots,r_s)$. Suppose that $\D_{\bs\alpha}(Re)\subseteq Q\subseteq \I$. We set $C=\{f\in \F(r_1,\dots,r_s)\tq f(\bs\alpha^m)=0,\;\text{for all }m\in Q\}$. Then $Q\subseteq \D_{\bs\alpha}(C)$, so that $\D_{\bs\alpha}(Re)\subseteq \D_{\bs\alpha}(C)$. We write $C=\sum_{k=1}^l Re_{i_k}$. By Corollary~\ref{suma de primitivos interseccion de conjuntos def} one has that $\D_{\bs\alpha}(C)=\cap_{k=1}^l \D_{\bs\alpha}(Re_{i_k})$; from which $\D_{\bs\alpha}(Re)\subset \D_{\bs\alpha}(Re_{i_k})$ for $k=1,\dots l$. Now suppose that there is a primitive idempotent  $e_{i_k}\neq e$. As they are primitive idempotents it must happen that $ee_{i_k}=0$. Since $e_{i_k}\neq 0$, by Lemma~\ref{polinomio cero caracterizado}, there exists $m\in \I$  such that $m\not\in\D_{\bs\alpha}(Re_{i_k})$, so that $e_{i_k}(\bs\alpha^m)\neq 0$ and $e(\bs\alpha^m)\neq 0$. A contradiction.

Conversely, suppose that $Q$ is maximal in $\Q$. Then, there exists $a\in \I$ such that $Q\cup Q(a)=\I$. Now, as we have pointed out, it happens that $\cap_{k=1}^t\D_{\bs\alpha}(Re_{i_k})=\emptyset$; so that there must exist a primitive idempotent, say $e=e_{i_k}$, such that $Q(a)\nsubseteq \D_{\bs\alpha}(Re)$ and by statement \textit{(1)} of this lemma, we have $Q(a)\cup \D_{\bs\alpha}(Re)=\I$. From here, $Q= \D_{\bs\alpha}(Re)$ and by Lemma~\ref{igualdad conj de definicion}, $\D_{\bs\alpha}(C)=\D_{\bs\alpha}(Re)$, so that $C\neq 0$. It remains to see that $C=Re$. Suppose that there exists $e'=e_{i_{k'}}$  with $k\neq k'$ such that $Re'\subset C$. Then $\D_{\bs\alpha}(Re')=\D_{\bs\alpha}(C)$, by maximality, and hence $\D_{\bs\alpha}(Re')=\D_{\bs\alpha}(Re)$, but $ee'=0$. As in the proof of paragraph above, we get a contradiction.
\end{proof}

Now we are going to prove the injectivity of the correspondence  $C\mapsto \D_{\bs\alpha}(C)\in \Q$. We begin by showing it for primitive idempotents. 

\begin{corollary}\label{inyectiva para los primitivos}
	Let $e,f\in \F(r_1,\dots,r_s)$ be primitive idempotents. If $\D_{\bs\alpha}(Re)=\D_{\bs\alpha}(Rf)$ then $e=f$.
\end{corollary}
\begin{proof}
	By Lemma~\ref{De es maximal} one has that $Re=Rf$ and, as they are primitive idempotents then $e=f$.
\end{proof}

For a subset $Q\subset \I$, we denote $\widehat{Q}=\I\setminus Q$; that is, its complement in $\I$. From results above it follows that if $e_1,\dots,e_t$ is the complete set of primitive idempotents of $\F(r_1,\dots,r_s)$ then it has exactly $m_1,\dots,m_t$ representatives of all distinct $q$-orbits such that $\widehat{\D_{\bs\alpha}(Re_i)}=Q(m_i)$ by reordering adequately all indexes.

\begin{proposition}\label{Correspondencia inyectiva}
	Let $C$ and $D$ be codes in $\F(r_1,\dots,r_s)$. If $C\varsubsetneq D$ then $\D_{\bs\alpha}(D)\varsubsetneq \D_{\bs\alpha}(C)$.
\end{proposition}
\begin{proof}
Let $e_1,\dots,e_t$ be the complete set of primitive orthogonal idemptents of $\F(r_1,\dots,r_s)$. Set $C=\sum_{j=1}^k Re_{i_j}$. By hypothesis, there is $f\in \{e_1,\dots,e_t\}$ such that $f\in C\setminus D$. If we call $Q(m_f)=\widehat{\D_{\bs\alpha}(Rf)}$, then we have, by  Corollary~\ref{inyectiva para los primitivos}, that $Q(m_f)\subset\D_{\bs\alpha}(Re_{i_j})$ for $j=1,\dots,k$ and hence $Q(m_f)\subset\cap_{j=1}^k\D_{\bs\alpha}(Re_{i_j})$; but $Q(m_f)\not\subset\cap_{j=1}^k\D_{\bs\alpha}(Re_{i_j})\cap \D_{\bs\alpha}(Rf)$, so that $\cap_{j=1}^k\D_{\bs\alpha}(Re_{i_j})\cap \D_{\bs\alpha}(Rf)\varsubsetneq \cap_{j=1}^k\D_{\bs\alpha}(Re_{i_j})$.
\end{proof}

Now we prove the injectivity in the general case.

\begin{corollary}\label{la correspondencia es inyectiva}
	Let $C$ and $D$ be codes in $\F(r_1,\dots,r_s)$. Then $C=D$ if and only if $\D_{\bs\alpha}(C)=\D_{\bs\alpha}(D)$.
\end{corollary}
\begin{proof}
	Necessity is trivial. Let us see sufficiency. Suppose we have $\D_{\bs\alpha}(C)=\D_{\bs\alpha}(D)$ and set $E=\{f\in \F(r_1,\dots,r_s)\tq f(\bs\alpha^m)=0,\;\text{for all }m\in \D_{\bs\alpha}(C)\}$. Then $C,D\subseteq E$ and by Lemma~\ref{igualdad conj de definicion} $\D_{\bs\alpha}(E)=\D_{\bs\alpha}(C)=\D_{\bs\alpha}(D)$. Finally, Proposition~\ref{Correspondencia inyectiva} tells us that $C=E=D$.
\end{proof}

Once injectivity has been proved, we are going to see that the correspondence $C\mapsto \D_{\bs\alpha}(C)\in \Q$ is onto. To do this, we have to prove that every element of $\Q$ is a defining set.

\begin{proposition}\label{la correspondencia es sobre}
	Let $Q\in\Q$ an arbitrary element. If $C=\{f\in \F(r_1,\dots,r_s)\tq f(\bs\alpha^m)=0,\;\text{for all }m\in Q\}$ then  $\D_{\bs\alpha}(C)=Q$. 
\end{proposition}
\begin{proof}
	Clearly $Q\subseteq \D_{\bs\alpha}(C)$. Let $m_1,\dots,m_k$ and $u_1,\dots,u_l$ be representatives of the $q$-orbits modulo $(r_1,\dots,r_s)$, such that $m_i\in Q$, for $i=\{1,\dots,k\}$ and $u_j\not\in Q$, for $j=\{1,\dots,l\}$. Then $Q=\cap_{j=1}^l\widehat{Q_q(u_j)}$. By Lemma~\ref{De es maximal}, each $\widehat{Q_q(u_j)}=\D_{\bs\alpha}(Re_j)$, with $e_j$ primitive idempotent, for $j=1,\dots,l$. So, $\D_{\bs\alpha}(\sum_{j=1}^l Re_j)=\cap_{j=1}^l\D_{\bs\alpha}(Re_j)=\cap_{j=1}^l\widehat{Q_q(u_j)}=Q$.
\end{proof}

\begin{theorem}\label{correspondencia biunivoca}
	In the setting and notation of this section, there exists a bijective correspondence between the abelian codes of $\F(r_1,\dots,r_s)$ and the elements of $\Q$, in the following way. For any code $C$ in $\F(r_1,\dots,r_s)=R$ and any  $Q\in\Q$,
	\begin{eqnarray*}
		C & \longrightarrow & \D_{\bs\alpha}(C)\\
		\{f\in R\tq f(\bs\alpha^m)=0,\;\text{for all }m\in Q\}& \longleftarrow & Q.
	\end{eqnarray*}
\end{theorem}

\begin{proof}
	Immediate from Corollary~\ref{la correspondencia es inyectiva} and Proposition~\ref{la correspondencia es sobre}.
\end{proof}

\section{Computing dimension through defining sets and the DFT}

We begin by comment on some well-known facts about tensor products and scalar extensions. Consider a field extension $\L|\F$. It is well-known (see \cite[Exercise 19.3]{anderson-fuller} or \cite[Corollary 1.7.16] {Rowen}) that if $V$ is an $\F$-vector space with basis $\{v_i\}_{i=1}^s$ then $\{1\otimes v_i\}_{i=1}^s$ is a basis for the $\L$-vector space $\L\otimes_\F V$ and as a consequence $\dim_\F V=\dim_\L(\L\otimes V)$.
	
In our setting, the set $\{\X^m\}_{m\in \I}$ is a basis for the $\F$-vector space $\F(r_1,\dots,r_s)$, so that $\{1\otimes \X^m\}_{m\in \I}$ is a basis  for $\L$-vector space $\L\otimes_\F \F(r_1,\dots,r_s)$. It is easy to check that the correspondence $1\otimes \X^m \mapsto \X^m$ induces an isomorphism of $\L$-vector spaces $\L\otimes_\F \F(r_1,\dots,r_s)\cong \L(r_1,\dots,r_s)$. Moreover, in \cite[Example 1.7.21.i]{Rowen}, it is proved that it is, in fact, an isomorphism of $\L$-algebras.

Following the arguments above one may prove that for any ideal $J\leq \F(r_1,\dots,r_s)$, it happens that $\dim_\F J=\dim_\L(\L\otimes_\F J)$. So, we have the following theorem.

\begin{theorem}\label{isomorfismo de algebras y dimension ideales}
	Let $\L|\F$ be a field extension. Then $\L\otimes_{\F}\F(r_1,\dots,r_s)\cong \L(r_1,\dots,r_s)$ as algebras. Moreover, if $J$ is a code in $\F(r_1,\dots,r_s)$  then $$\dim_{\F}(J)=\dim_{\L}(\L\otimes_{\F} J).$$
\end{theorem}

Now we review the definition and basic properties of the discrete Fourier transform (DFT, for short) in the context of multivariate polynomials. The reader may see \cite{camion} for a presentation in the context of character theory.

Following the notation settled in Section~\ref{seccion de notacion}, let $f \in \F(r_1,\dots,r_s)$ be an arbitrary element, that we consider as a polynomial. The DFT with respect to $\bs\alpha\in U$ is the polynomial $\varphi_{\bs{\alpha},f}(\X)=\sum_{m\in \I}f(\bs{\alpha}^m)\X^m \in \L(r_1,\dots,r_s)$.

We want to formalize this concept and view it as a map. The first important fact that we recall is that, for any $f \in \F(r_1,\dots,r_s)$, its image verifies $\varphi_{\bs{\alpha},f}(\X)\in \L(r_1,\dots,r_s)$; that is, it belongs to the polynomial factor ring over the extension field, $\L$. Now with respect to arithmetical properties, it is easy to check that, for $f,g \in \F(r_1,\dots,r_s)$, we have $\varphi_{\bs{\alpha},f+g}(\X)=\varphi_{\bs{\alpha},f}(\X)+\varphi_{\bs{\alpha},g}(\X)$; however, the DFT is not multiplicative with respect to the usual polynomial product. This is the reason for which it is considered another multiplication on the codomain. For $f,g \in \L(r_1,\dots,r_s)$, with $f=\sum_{m\in \I}f_m\X^m$ and $g=\sum_{m\in \I}g_m\X^m$ we define $f\star g = \sum_{m\in \I}f_m g_m\X^m$; that is, the product coeficient by coeficient or coordinatewise. We denote this algebra by $(\L^n,\star)$. Now, clearly $\varphi_{\bs{\alpha},f\star g}(\X)= \varphi_{\bs{\alpha},f}(\X)\star \varphi_{\bs{\alpha},g}(\X)$.

One may prove that the DFT, viewed as $\varphi:\L(r_1,\dots,r_s)\flecha (\L^n,\star)$ is an isomorphism of algebras with inverse $\varphi^{-1}_{\bs{\alpha},f}(\X)=\frac{1}{s}\sum_{m\in \I}f(\bs{\alpha}^{-m})\X^m$.
	
In the following lemma, we find three important facts on the DFT that are easy to check.

\begin{lemma}\label{las tres propiedades de la DFT}
 In the setting  above, with $\varphi_{\bs\alpha,-}:\L(r_1,\dots,r_s)\flecha (\L^n,\star)$, let $e\in \F(r_1,\dots,r_s)$ and $f\in (\L^n,\star)$ be idempotent elements in their respective algebras. Then
 \begin{enumerate}
  \item \label{idempotentes en L} If $f=\sum_{m\in I}f_m\X^m$ then we must have $f_m\in \{0,1\}$.
 % \item \label{coincidencia conjuntos de def en F y L} The equality $\D_{\bs\alpha}(\F e)=\D_{\bs\alpha}(\L e)$ holds.
  \item \label{soporte de Phi e} The support $\supp\left(\varphi_{\bs{\alpha},e}\right)=  \I\setminus \D_{\bs\alpha}(\F(r_1,\dots,r_s) e)$.
 \end{enumerate}
\end{lemma}

We are now ready to show the following easy formula to compute the dimension of abelian codes, expressed only in polynomial terms.

\begin{theorem}\label{dimension codigos conj def}
Let $\bs\alpha \in U$ be fixed and let $C$ an abelian code in $\F(r_1,\dots,r_s)$. Then $$\dim_\F(C)= |\I\setminus \D_{\bs\alpha}(C)|=n-|\D_{\bs\alpha}(C)|.$$
\end{theorem}
\begin{proof}
Let $e\in \F(r_1,\dots,r_s)$ the generating idempotent of $C$. Then $C= \F(r_1,\dots,r_s) e$, and clearly $\L\otimes_\F C=\L(r_1,\dots,r_s) e$. By Theorem~\ref{isomorfismo de algebras y dimension ideales}, $\dim_F(C)=\dim_\L(\L(r_1,\dots,r_s) e)$. If we call $K$ the ideal generated by $\varphi_{\bs{\alpha},e}$ in $(\L^n,\star)$ then, we see immediately from Lemma~\ref{las tres propiedades de la DFT} that $\dim_\L(K)=|\supp\left(\varphi_{\bs{\alpha},e}\right)|$.

Now, as the DFT is an isomorphism of algebras then $\dim_\L(\L(r_1,\dots,r_s) e)=\dim_\L(K)$. Finally,
\begin{eqnarray*}
 \dim_\F(C)&=&\dim_\L(\L(r_1,\dots,r_s) e)=\\ &=&\dim_\L(K)=|\supp\left(\varphi_{\bs{\alpha},e}\right)|= \\
 &=& |\I\setminus \D_{\bs\alpha}(C)|=n-|\D_{\bs\alpha}(C)|.
\end{eqnarray*}
\end{proof}

As a consequence, we may reformulate the computing of the dimension in terms of the DFT.

\begin{corollary}
Let us fix an element $\bs\alpha \in U$ and let $C$ be an abelian code in $\F(r_1,\dots,r_s)$ with generating idempotent $e\in C$. Then $\dim_\F(C)=|\supp\left(\varphi_{\bs{\alpha},e}\right)|$.
\end{corollary}

An interesting consequence of the relationship between defining sets and the DFT is the following. Let $\bs\alpha \in U$ be fixed and let $C$ an abelian code in $\F(r_1,\dots,r_s)$ with defining set $\D_{\bs\alpha}(C)$. We may compute explicitly the generator idempotent of $C$, as $e=\frac{1}{s}\sum_{m\in \I\setminus \D_{\bs\alpha}(C)}f(\bs{\alpha}^{-m})\X^m$.

\section{Computing dimension through groebner basis}

Groebner bases are an important tool in the study of abelian codes. As in the case of defining sets, essential aspects of such codes, as information sets, bounds for the minimum distance and coding or decoding techniques, may be study by means of Groebner bases and their footprints. The reader may see \cite{chabanne,Sala y Mora} for more information.

In this section we see another interesting application of Groebner bases in the study of Abelian Codes. We shall show that, one may compute the dimension of any Abelian Code, $C$ by means of the so-called footprint of $C$.

Let us recall some notation and results about Groebner bases. Essentially, we shall follow \cite[Chapter 2]{Cox} and \cite{blahut}. So, we denote
\[\Z^s_{\geq 0}=\{(m_1,\dots,m_s)\in \Z^s\tq m_i\geq 0,\;\forall\,i\in \{1,\dots,s\}\}.\]

\begin{definition}
Consider a monomial ordering $\leq_T$ for $\Z^s_{\geq 0}$ and an element $f\in \F[\X]$, that we write $f=\sum_{m\in \supp(f)}f_m\X^m$.
 \begin{enumerate}
 \item The leading exponent, or multidegree of $f$ is  $\LP(f)=\max_{\leq_T}\{m\in \Zs\tq m\in\supp(f)\}$.
  \item The leading term is $\LT(f)=f_{\LP(f)}\X^{\LP(f)}$.
  \item For any subset  $J\subset \F[\X]$ we define the set of leading terms as $\LT(J)=\{\LT(f)\tq f\in J\}$.
 \end{enumerate}
\end{definition}
 
From Dickson's lemma \cite[Theorem 2.4.5]{Cox}, the Hilbert's basis theorem \cite[Theorem 2.5.4]{Cox} and other results one may deduce the existence and properties of Groebner basis.
 
 \begin{definition}
   Let $\leq_T$ be a monomial ordering for $\Z^s_{\geq 0}$ and let $\{0\}\neq J\leq  \F[\X]$ an ideal. We say that the set $\{g_1,\dots,g_t\}$ is a Groebner basis for $J$ if
   \[\langle \LT(g_1),\dots, \LT(g_t)\rangle=\langle \LT(J)\rangle.\]
   
   By convention, $\langle\emptyset\rangle=\{0\}$ and so $\emptyset $ is a Groebner basis for $\{0\}$
 \end{definition}

\begin{remark}
It is immediate to prove that, if $\G=\{g_1,\dots,g_t\}$ is a Groebner basis for $J$ then for all $f\in J$ there exists $g_i\in\G$ such that, its leading term divides to that of $f$; that is, $\LT(g_i)\mid \LT(f)$.
\end{remark}

From here we consider the following partial order in $\Z_{\leq 0}$. 
 
 \begin{definition}
For any pair $m,m'\in\Zs$ with $m=(m_1,\dots,m_s)$ and $m'=(m'_1,\dots,m'_s)$ we say that $m\preceq m'$ if $m_i\leq m_i'$ for each $i=1,\dots,s$. 
 \end{definition}
 
 It is known that (see \cite[Paragraph 2.6]{Cox}), the division algorithm by elements of a Groebner basis has a good behavior. We summarize this in the next result.
 
 \begin{proposition}
  Let $\G=\{g_1,\dots,g_t\}$ be a Groebner basis for $J\leq \F[\X]$ and $f\in\F[\X]$, arbitrary. Then, there exists a unique division remainder, $r\in \F[\X]$ satisfying the following properties:
  \begin{enumerate}
   \item No term of $r$ is greater than or equal to any of $\LT(g_1),\dots,\LT(g_t)$.
   \item $f=g+r$ with $g\in J$. In particular, $f\in J$ if and only if $r=0$.
  \end{enumerate}
 \end{proposition}
 
 \begin{definition}
  The remainder $r$ is called the normal form of $f$.
  
  We say then that $f$ is in normal form, with respect to $\G$, in case $f=r$.
 \end{definition}

So, the result above says that if $\G=\{g_1,\dots,g_t\}$ is a Groebner basis for $J$ then, for all $f\in J$ there exists $g_i\in\G$ such that $\LT(g_i)\preceq \LT(f)$.

\begin{definition}
Let $J\leq \F[\X]$ an ideal with Groebner basis $\G$, such that all its elements are monic.
 \begin{enumerate}
  \item We say that $\G$ is minimal if, for all $g\in \G$ we have that $\LT(g)\not\in \langle \LT(\G)\rangle$.
  \item We say that $\G$ is reduced if for all $g\in \G$ no monomial of $g$ belongs to $\langle \LT(\G)\setminus \{g\}\rangle$ (that is, the elements of $\G$ are in normal form).  
 \end{enumerate}
 \end{definition}

We know by \cite[Theorem 2.7.5]{Cox} that every nonzero ideal of $\F[\X]$ has a unique reduced Groebner basis.
 
 \subsection{Delta sets and the footprint of a code. Dimensions}
 
Now we shall apply the results above to the computation of the dimension of Abelian Codes. Set $q,s$ and $r_1,\dots,r_s$, with $n=\prod_{i=1}^s r_i$, as above. We start with the following result from elementary algebra.

Consider the ideal $K=\langle X^{r_1}-1,\dots,X^{r_s}-1\rangle$ and the quotient ring $\F[X_1,\dots,X_s]/K=\F[\X]/K=\F(r_1,\dots,r_s)$. Let $C\leq \F(r_1,\dots,r_s)$ be a code. By the Correspodence Theorem, we know that there is a unique ideal $J\leq \F[\X]$ such that $K\subseteq J$ and that $J/K=C$. Now, by the Second Isomorphism Theorem,
\[\F[\X]/J\cong \left(\F[\X]/K\right)/\left(J/K\right)\cong \F(r_1,\dots,r_s)/C\]
is an isomorphism of algebras; and so as $\F$-vector spaces we have that
\begin{equation}\label{isomorfismo de algebras}
 \dim_\F(\F[\X]/J)=\dim_\F(\F(r_1,\dots,r_s)/C)=n-\dim_\F(C).
\end{equation}

Now, let $C\leq \F(r_1,\dots,r_s)$ be a code and let  $J\leq \F[\X]$ be the unique ideal such that $K\subseteq J$ and $J/K=C$. Fix a monomial ordering $\leq_T$. Let $\G_{\leq_T}=\G=\{g_1,\dots,g_t\}$ a minimal Groebner basis for $J$. It is clear that, as  $K\subset J$ there exists a subset $\{f_1,\dots,f_s\}\subset\G$, such that $\LP(f_i)$ is such that, $\LP(f_i)\preceq (0,\dots,0,r_i,0\dots,0)$.
  
Following the notation, we define
\[A(\G)= \sum_{i=1}^t \LP(g_i)+\Zs.\]
By the properties of minimal Groebner basis described above, we know that if $\G'$ is another minimal basis of $J$ \textbf{under the same ordering} ``$\leq_T$'', then $A(\G)=A(\G')$. This uniqueness allows us to give the notion of the footprint.

\begin{definition}
Let $\leq_T$ be a fixed monomial ordering and $C\leq \F(r_1,\dots,r_s)$ be a code and let  $J\leq \F[\X]$ be the unique ideal such that $K\subseteq J$ and $J/K=C$. We call the footprint of $C$ the set 
\[\Delta(C)=\Zs \setminus A(\G),\]
where $\G$ is any minimal Groebner basis for $J$.
\end{definition}

Let us recall two elementary properties related to the results above. First, by definition of Groebner basis, we have that, for $f\in \F[\X]$, with normal form $r_f$ it happens that if $f\not\in J$ then $0\neq \LP(r_f)\in \Delta(J)$. Second, by the uniqueness of the remainder we have that, for $f\in \F[\X]$, the correspondence $f\mapsto r_f$ is a $\F$-linear tranformation.
 
On the other hand, note that the footprint of any Abelian Code, $C$, is a finite set; in fact, $\Delta(C)\subseteq \I$.

\begin{theorem}\label{dimension bases de groebner}
 Let $\leq_T$ be a fixed monomial ordering and $C\leq \F(r_1,\dots,r_s)$ be an Abelian Code with footprint $\Delta(C)$. Then by taking the canonical representatives, the set  $\{\X^m\tq m\in \Delta(C)\}$ is a basis for $\F(r_1,\dots,r_s)/C$.
 
 Hence, $\dim_\F(C)=n-|\Delta(C)|$.
\end{theorem}
\begin{proof}
First, note that, from the definition of minimal Groebner  basis together with the fact that, for any $m\in\Delta(C)$, the normal form of the polynomial $\X^m$ is $ r_{\X^m}=\X^m$. This means that the own $\X^m$ is a canonical representative. So, the set $\{\X^m\tq m\in \Delta(C)\}$ is a $\F$-basis for $\F[\X]/J$. Now, from the isomorphism in Equation~\eqref{isomorfismo de algebras} $\F[\X]/J\cong F(r_1,\dots,r_s)$ the result follows immediately. The last part follws directly from the fact that the correspondence $f\mapsto r_f$ is $\F$-linear.
\end{proof}

Combining Theorem~\ref{dimension codigos conj def} and Theorem~\ref{dimension bases de groebner}, above, we get the following corollary.

\begin{corollary}\label{huella y conj def}
 In the setting of Theorem~\ref{dimension codigos conj def} and Theorem~\ref{dimension bases de groebner} we have that $|\Delta(C)|=|\D_\alpha(C)|$.
\end{corollary}

\section{Interactions between defining sets and footprints.}

As we have commented above, some of the most important parameters, structure and coding-decoding techniques can be done by means of both defining sets through the discrete Fourier transform and the footprint through minimal Groebner basis.

In practice, the decision to use one tool or another, depends on the code construction; for instance, algebraic geometric codes make more use of footprints than general abelian codes.

We remark that in any case (if necessary) it is always possible to move from defining sets to footprints and vice versa. This transit can be done even avoiding the previous calculation of a minimal Groebner basis, as it is shown in \cite{BS2}.

We conclude by pointing out that there are frameworks in which all the tools that we have seen in this note must be combined to achieve a goal. Perhaps the most beautiful example of this is locator decoding (see \cite{BS3,sakata2}). We need the defining set of the code to compute the syndrome values; then we implement the Berlekamp-Massey-Sakata algorithm to get a minimal Groebner basis of the locator ideal and finally we get the defining set of the locator ideal to get the error locations. By Corollary~\ref{huella y conj def} we have that the number of errors occurred during transmission is exactly the cardinality of the footprint of the locator ideal.

\end{document}